\documentclass[conference, a4paper]{IEEEtran}

\usepackage[colorinlistoftodos]{todonotes}
\usepackage[utf8]{vietnam}
\usepackage{array}
\usepackage{lipsum}
\usepackage{subfigure}
\usepackage{listings}
\usepackage{xcolor}
\usepackage{amsthm}

\newtheorem{theorem}{Định lí}
\usepackage{algorithm}
\usepackage{algpseudocode}

\IEEEoverridecommandlockouts
\usepackage{cite}
\usepackage{amsmath,amssymb,amsfonts}
\usepackage{graphicx}
\usepackage{textcomp}
\usepackage{xcolor}
\def\BibTeX{{\rm B\kern-.05em{\sc i\kern-.025em b}\kern-.08em
    T\kern-.1667em\lower.7ex\hbox{E}\kern-.125emX}}

\begin{document}

\title{
Energy efficient deployment solutions in high density heterogeneous networks}

\author{

\IEEEauthorblockN{
Nguyễn Doãn Hiếu* \IEEEauthorrefmark{2},
Đào Lê Thu Thảo \IEEEauthorrefmark{2},
Trần Mạnh Hoàng \IEEEauthorrefmark{2}
} 
\IEEEauthorblockA{\IEEEauthorrefmark{2}\textit{Viện Điện tử - Viễn thông, Trường Đại học Bách khoa Hà Nội}\\
Email: \{hieu.nd172526\}@sis.hust.edu.vn; \{thao.daolethu, hoang.tranmanh\}@hust.edu.vn}
}

\maketitle

\begin{abstract} 

This study deals with the problem of optimizing transmit power in high density heterogeneous networks. In the communication network, effective methods of allocating transmit power, in order to reduce the total transmit power, but still ensure the quality of service of the user equipment, is a big challenge. number of power consumption optimization problems in core station links, with the goal of maximizing network energy efficiency while ensuring user experience. To solve this non-convex optimization problem, the authors first propose some iterative algorithms to find the convergence point such as the "descent" method, the "Lagrange" method. Then, the authors evaluate the convergence point of each method as well as consider the complexity of each algorithm when put into application. Finally, the simulation results will show the convergence value and compare the performance with the technologies being used today to confirm the effectiveness of the proposed algorithm.
\end{abstract}
\def\IEEEkeywordsname{Từ khóa}

\section{Giới thiệu}

Hệ thống thông tin di động không dây 5G dự báo lưu lượng truy cập cũng như số lượng thiết bị di động tăng cực mạnh. Mạng không đồng nhất (HetNets) được xem là giải pháp tốt nhất có thể đáp ứng nhu cầu của người dùng bằng cách sử dụng nhiều lớp mạng gồm macro/pico/femtocell để phủ sóng tốt hơn cũng như nâng cao hiệu suất phổ (SE) và hiệu suất năng lượng (EE) \cite{Huq2014}. 

Nhu cầu sử dụng các dịch vụ truyền thông ngày càng lớn, đo đó để đáp ứng nhu cầu này, hệ thống mạng phải triển khai một số lượng lớn các cell cỡ nhỏ (small cells) trong các mạng di động thông thường và gây ra tình trạng mạng không đồng nhất – mật độ cao (Ultra-dense heterogeneous networks HetNets) \cite{Kamel2016}. Bằng cách thêm các loại trạm phát công suất thấp (Base Station - BS) khác nhau, ví dụ: pico, femto và các trạm chuyển tiếp, HetNets có thể mang lại hiệu suất tăng đáng kể, chẳng hạn như hiệu suất quang phổ, hiệu suất công suất và phạm vi phủ sóng đầy đủ \cite{Thanh2021}. Để cải thiện hơn nữa hiệu suất mạng, người dùng di động (User Equipment - UE) thường kết nối với BS tế bào nhỏ, thậm chí BS cỡ lớn có thể cung cấp tín hiệu di động mạnh nhất, tức là giảm tải hoặc cân bằng tả \cite{Nguyen2021}. Một sơ đồ tối ưu kênh truyền, phù hợp giúp giảm tải trên các cell cỡ lớn (macro cell), giảm công suất phát trung bình và cải thiện tốc độ dữ liệu dài hạn của các UE do khoảng cách giữa các UE và BS nhỏ hơn \cite{Tam2017}. Tuy nhiên, một trong những vấn đề chính trong HetNets mật độ cao là quản lý nhiễu, điều này có thể làm giảm hiệu quả của việc triển khai năng lượng. Do đó, cần phải xem xét về việc điều khiển công suất và mô hình quản lý nhiễu (Traffic offloading and interference management - TOIM) để có được lợi ích cao hơn.

Trong \cite{Sun2015}, các tác giả trình bày một phân tích về độ phức tạp và hiệu quả của thuật toán được nghiên cứu trong bài toán công bằng tối thiểu của mô hình quản lý nhiễu. Các tác giả trong \cite{Quan2019} đã xem xét hai loại mục tiêu, mục tiêu thứ nhất là tỷ lệ hiệu dụng tổng mạng và mục tiêu thứ hai là tỷ lệ hiệu dụng tối thiểu, đồng thời phát triển hai thuật toán tương ứng bằng cách sử dụng phương pháp xấp xỉ lồi liên tiếp và suy giảm xen kẽ (Successive convex approximation). Tuy nhiên, lưu lượng dữ liệu ngày càng phát triển, gây ra sự gia tăng đáng kể mức tiêu thụ năng lượng, đặc biệt là trong các HetNets siêu dày đặc, nơi một số lượng lớn các cell cỡ nhỏ được triển khai để phục vụ một số lượng lớn các UE. Điều này luôn thúc đẩy nhóm tác giả làm việc theo hướng nâng cao hiệu quả sử dụng năng lượng mạng. Trong \cite{Agapi2014}, tác giả giới thiệu một phương pháp liên kết thiết bị người dùng, có nhận biết thông tin ngoại cảnh, được lập trình để nâng cao hiệu suất năng lượng cho HetNets. Phương pháp được đề xuất hoạt động theo phương thức khai thác khả năng nhận thức của các phần tử mạng, tức là HetNets có nhận thức; tuy nhiên, bài báo không xem xét các vấn đề về nhiễu và công suất tiêu thụ của các liên kết trạm lõi - liên kết backhaul. Các tác giả trong \cite{Wang2017} đã nghiên cứu sơ đồ mô hình quản lý nhiễu tiết kiệm năng lượng trong tuyến lên của HetNest, nơi một UE có thể kết nối với macro BS ở chế độ trực tiếp hoặc chế độ chuyển tiếp. 

Dựa trên các quan sát trên, tác giả xem xét một vấn đề chung về nâng cao hiệu suất năng lượng tuyến xuống của HetNets mật độ cao, trong đó cả công suất tiêu thụ của mạng truy cập và liên kết trạm lõi đều được xem xét trong mục tiêu của nâng cao hiệu suất năng lượng. Giải pháp tối ưu cho vấn đề đã được giải quyết bằng cách sử dụng chung quy trình Dinkelbach’s và một số phương pháp tối ưu cho bài toán tối ưu lồi. Cuối cùng, tác giả lựa chọn phương pháp suy giảm luân phiên sau khi đã cân nhắc về hiệu suất công suất mạng cũng như độ phức tạp, chi phí hệ thống vận hành. Kết quả mô phỏng số được cung cấp để đánh giá hiệu suất của thuật toán được đề xuất so với các phương thức quản lý mạng hiện có.

Cấu trúc bài báo sẽ gồm 4 phần: trong phần \ref{Sec:ProblemFomulation}, nhóm tác giả  mô tả bài toán truyền tin tuyến xuống trong mạng HetNets. Trong phần \ref{Sec:ProposedAlgorithm}, các tác giả đề xuất một số thuật toán giải quyết vấn đề đã đặt ra. Cuối cùng, kết quả được thu thập  và đánh giá trong phần \ref{Sec:SimulationResults}

\section{Mô hình hệ thống và đặt vấn đề}
\label{Sec:ProblemFomulation}

Trong nghiên cứu này, nhóm tác giả xem xét một mô hình hệ thống mạng không đồng nhất như Hình \ref{Fig:HetNet}. Trong mạng truyền thông không đồng nhất HetNet 2 tầng này gồm 1 cell cỡ lớn và $K-1$ cell cỡ nhỏ. Để đạt được hiệu suất tái sử dụng phổ cao, các cell cỡ nhỏ được kích hoạt để tái sử dụng toàn bộ phổ của cell vi bào. Ký hiệu $\mathcal{K}=\left\lbrace 1,...,K \right\rbrace$ là tập hợp các trạm phát trong đó chỉ số của trạm phát cell cỡ lớn (Macrocell base station -MBS) là $1$, kí hiệu $\mathcal{N}=\left\lbrace 1,...,N \right\rbrace$ là tập hợp thiết bị người dùng. 
\begin{figure} [ht]
    \centering
    \includegraphics[scale=0.23]{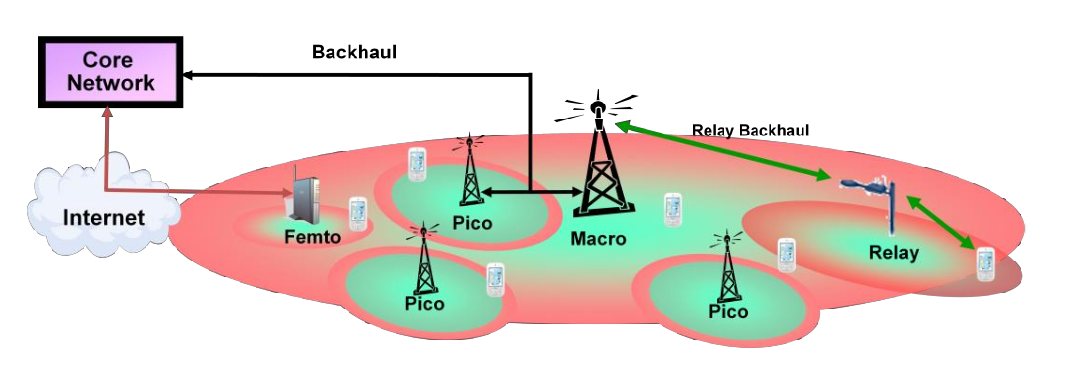}
    \caption{Mô hình một mạng không đồng nhất}
    \label{Fig:HetNet}
\end{figure}
Tỉ số tín hiệu trên tạp âm và nhiễu (SINR) của UE thứ $n$, đang kết nối tới BS thứ $k$, là 
\begin{align*}
    \gamma_{nk}(\boldsymbol{p}) = \frac{g_{nk}p_{k}}{\left(n_{0} + \sum\nolimits_{j \neq k, j \in \mathcal{K}}g_{nj}p_{j}\right)}.
\end{align*} 

Trong đó tỉ lệ tín hiệu người dùng thứ $n$ thu được là $r_{nk}(\boldsymbol{p}) = \log(1 + \gamma_{nk}(\boldsymbol{p}))$. Với $p_{k}$ là công suất truyền tín hiệu của BS thứ $k$, $g_{nk}$ là tăng ích kênh truyền giữa BS $k$ và UE $n$, và $n_{0}$ là công suất nhiễu nền tại UE thứ $n$. Kí hiệu $x_{nk}$ là biến kết hợp của UE thứ $n$, cụ thể, $x_{nk}=1$ thể hiện UE thứ $n$ kết nội với BS thứ $k$, còn $x_{nk}=0$ thì ngược lại. Tổng số lượng UE kết nối với BS thứ $k$ là $y_{k}=\sum\nolimits_{n \in \mathcal{N}}x_{nk}$. Tiếp theo, tỉ lệ tín hiệu hiệu dụng của UE thứ $n$ được định nghĩa $R_{nk}(\boldsymbol{p},\boldsymbol{x}_{k})=r_{nk}(\boldsymbol{p})/y_{k}$ \cite{Tam2017}, trong đó $\boldsymbol{x}_{k}=[x_{1k},...,x_{Nk}]^{T}$ và $\boldsymbol{x}=[\boldsymbol{x}_{1};...;\boldsymbol{x}_{K}]$. Khái niệm về tỷ lệ tín hiệu hiệu dụng có thể đạt được của một UE tương tự như tỷ lệ có ích và tỉ lệ có ích hiệu dụng trong các khuôn khổ kiểm soát tắc nghẽn xuyên lớp \cite{QVP2015}. Dựa trên tỷ lệ hiệu dụng của mỗi UE, tổng tỷ lệ hiệu dụng của mạng là
\begin{equation}
    R(\boldsymbol{x},\boldsymbol{p}) = \sum\nolimits_{k \in \mathcal{K}}\sum\nolimits_{n \in \mathcal{N}}x_{nk}R_{nk}(\boldsymbol{p},\boldsymbol{x}_{k}), \nonumber
\end{equation}
trong đó $x_{nk}$ thể hiện tỉ lệ hữu ích của người dùng thứ $n$ thu được nếu người dùng thứ n kết hợp với BS thứ $k$.

Tổng công suất tiêu hao là tổng công suất trong mạng kết nối và trong liên kết trạm lõi \cite{Agapi2014}, với
\begin{equation}
    P = P_{\rm an} + P_{\rm bh} \nonumber,
\end{equation}
trong đó $P_{\rm an}$ là công suất kết nối giữa các anten, được biểu diễn như sau:
\begin{equation}
    P_{\rm an} = \sum\nolimits_{k \in \mathcal{K}}\varrho_{k}p_{k} + P_{c} \nonumber,
\end{equation}
và $P_{\rm bh}$ là công suất của các liên kết trạm lõi. $\varrho_{k}$ là nghịch đảo của hiệu suất tiêu hao của bộ khuếch đại công suất tại BS thứ $k$ và $P_{c}$ công suất tiêu thụ cố định. Công suất tiêu thụ trong một liên kết trạm lõi - liên kết backhaul - thường tỷ lệ với tổng dữ liệu được mang bởi liên kết đó $P_{\rm bh} = \sum\nolimits_{k \in \mathcal{K}}(\xi_{k}/y_{k})\sum\nolimits_{n \in \mathcal{N}}x_{nk}r_{nk}$, trong đó $\xi_{k}$ tính mức tiêu thụ năng lượng động trên mỗi đơn vị dữ liệu \cite{isheden2010}.

Nghiên cứu nhắm tới việc tìm ra liên kết người dùng tối ưu và phân bổ công suất để  nâng cao hiệu suất năng lượng của mạng, được tính bằng tỉ lệ $EE = R/P$. Vấn đề tối ưu hóa có thể được lập công thức toán học như
\begin{align}
& \max_{\boldsymbol{x},\boldsymbol{p}} \left[ EE = \frac{\sum\limits_{k \in \mathcal{K}}\sum\limits_{n \in \mathcal{N}}x_{nk}r_{nk}/y_{k}}{\sum\nolimits_{k \in \mathcal{K}}\varrho_{k}p_{k} + P_{c} + \sum\limits_{k \in \mathcal{K}}(\xi_{k}/y_{k})\sum\limits_{n \in \mathcal{N}}x_{nk}r_{nk}}  \right] \label{Eq:OptProblem}\\
& \text{s.t.} 
\quad (\text{C1}): \sum\limits_{k \in \mathcal{K}}x_{nk} = 1, \; \forall n \in \mathcal{N}	\label{Eq:OptProblem} \nonumber \\
& 
\quad\quad\: (\text{C2}): y_{k} = \sum\limits_{n \in \mathcal{N}}x_{nk} \geq 1, \; \forall k \in \mathcal{K} \nonumber \\
&
\quad\quad\: (\text{C3}): x_{nk} = \left\lbrace 0,1 \right\rbrace, \; \forall n \in \mathcal{N}, \; \forall k \in \mathcal{K} \nonumber \\
&
\quad\quad\: (C4): 0 \leq p_{k} \leq p_{k}^{\max}, \; \forall k \in \mathcal{K} \nonumber .
\end{align}
Ràng buộc (C1) thể hiện mỗi UE được phép kết nối với nhiều nhất một BS. Ràng buộc (C2) cho rằng mỗi BS phục vụ ít nhất một UE. Trong ràng buộc (C3), biến kết hợp $x_{nk}$ có thể là $0$ hoặc $1$. Cuối cùng, ràng buộc (C4) giới hạn của mỗi BS thứ $k$ bởi công suất phát tối đa $p_{k}^{\max}$. Gọi tập hợp nghiệm khả thi là $\mathcal{F}$.

Quan sát ta có,~\eqref{Eq:OptProblem} là một bài toán phân số phi tuyến tính hỗn hợp số nguyên; do đó, nó thực sự khó giải quyết trong một hàm đa thức thời gian. Bởi vậy, các tác giả đề xuất sử dụng quy trình Dinkelbach’s và khung tối ưu hóa xen kẽ. Đặc biệt, nhóm tác giả chuyển đổi bài toán \eqref{Eq:OptProblem} thành một dạng trừ, được tách thành một chuỗi bài toán điều khiển công suất cho một liên kết người dùng cố định và một bài toán liên kết người dùng cho một phân bổ công suất cố định.

\section{Đề xuất thuật toán}
\label{Sec:ProposedAlgorithm}
Xác định bằng $q^{*}$ hiệu suất năng lượng lớn nhất đạt được với nghiệm tối ưu$(\boldsymbol{p}^{*},\boldsymbol{x}^{*})$,, tức là $q^{*} = R(\boldsymbol{p}^{*},\boldsymbol{x}^{*})/P(\boldsymbol{p}^{*},\boldsymbol{x}^{*})$. Định lý sau xác định các điều kiện để đạt được lời giải tối ưu cho bài toán~\eqref{Eq:OptProblem}.

\begin{theorem}
Kết quả tối ưu $(\boldsymbol{p}^{*},\boldsymbol{x}^{*})$ đạt được hiệu suất năng lượng tối ưu  $q^{*}$ khi và chỉ khi
\begin{equation}\label{Eq:Dkb_Sub_OptProblem}
\begin{split}
\max\limits_{\boldsymbol{p},\boldsymbol{x} \in \mathcal{F}}\left[ R(\boldsymbol{p},\boldsymbol{x}) - q^{*}P(\boldsymbol{p},\boldsymbol{x}) \right] &= R(\boldsymbol{p}^{*},\boldsymbol{x}^{*}) - q^{*}P(\boldsymbol{p}^{*},\boldsymbol{x}^{*}) \nonumber \\ 
&= 0.
\end{split}
\end{equation}
Do đó, nếu biết trước $q^{*}$, chúng ta có thể tìm lời giải tối ưu cho bài toán~\eqref{Eq:OptProblem} bằng cách giải tương đương bài toán sau:
\begin{equation} \label{Eq:Sub_OptProblem}
\max\limits_{\boldsymbol{p},\boldsymbol{x} \in \mathcal{F}}\left[ R(\boldsymbol{p},\boldsymbol{x}) - q^{*}P(\boldsymbol{p},\boldsymbol{x}) \right].
\end{equation}
\end{theorem}
Tuy nhiên, giá trị tối ưu $(\boldsymbol{p}^{*}$ thường không được xác định trước. Do  đó,tác giải đề xuất một thuật toán, trong đó sẽ giải quyết vấn đề \eqref{Eq:Sub_OptProblem} với $(\boldsymbol{p}$ cho trước và giá trị của $(\boldsymbol{p}$ được thay đổi sau mỗi lần lặp lại. Bây giờ, bài toán \eqref{Eq:Sub_OptProblem} được viết lại như sau:
\begin{equation} \label{Eq:SubEq_OptProblem}
\max\limits_{\boldsymbol{p},\boldsymbol{x} \in \mathcal{F}} \sum\limits_{k \in \mathcal{K}}\sum\limits_{n \in \mathcal{N}}\frac{1 - q\xi_{k}}{y_{k}} x_{nk}r_{nk} - q\sum\limits_{k \in \mathcal{K}}\varrho_{k}p_{k}.
\end{equation}

\subsection{Kết hợp người dùng và cố định phân bổ công suất}
Đối với phân bổ công suất đã cho $\boldsymbol{p}$, bây giờ chúng ta giải quyết vấn đề \eqref{Eq:SubEq_OptProblem} đối với (w.r.t.) vectơ kết hợp $\boldsymbol{x}$. Trong trường hợp này, bài toán  \eqref{Eq:SubEq_OptProblem} tương đương với:
\begin{align} \label{Eq:UEAssociation_Problem}
& \max\limits_{\boldsymbol{x}} \sum\limits_{k \in \mathcal{K}}\sum\limits_{n \in \mathcal{N}}\frac{1 - q\xi_{k}}{y_{k}} x_{nk}r_{nk}
\\
& \text{s.t.} 
\quad (\text{C1}), (\text{C2}), (\text{C3}) \nonumber .
\end{align}
Để giải bài toán \eqref{Eq:UEAssociation_Problem}, chúng ta đưa ra các biến phụ $t_{nk}$ như là $r_{nk} \geq \sum\nolimits_{n \in \mathcal{N}}x_{nk}t_{nk}$. Bài toán \eqref{Eq:UEAssociation_Problem} có thể được viết lại tương đương như:
\begin{align} \label{Eq:UEAssociation_Aux_Problem}
& \max\limits_{\boldsymbol{x}} \sum\limits_{k \in \mathcal{K}}\sum\limits_{n \in \mathcal{N}}\left( 1 - q\xi_{k} \right) x_{nk}t_{nk}
\\
& \text{s.t.} 
\quad (\text{C1}), (\text{C2}), (\text{C3}) \nonumber, \\
&
\quad\quad\: r_{nk} \geq y_{k}t_{nk}, \; \forall n \in \mathcal{N}, \; \forall k \in \mathcal{K} \nonumber.
\end{align}
Định lý sau đây cho phép chúng ta đúc kết lại bài toán~\eqref{Eq:UEAssociation_Aux_Problem} một cách tương đương như một bài toán có khả năng giải trừ.

\begin{theorem}	\label{Theo:UEAssociation_Sub_Problem}
Coi $(\boldsymbol{x}^{*},\boldsymbol{t}^{*})$ là nghiệm tối ưu của bài toán~\eqref{Eq:UEAssociation_Aux_Problem}, tồn tại $\boldsymbol{\lambda}=[\boldsymbol{\lambda}_{1};...;\boldsymbol{\lambda}_{K}]$, trong đó $\boldsymbol{\lambda}_{k}=[\lambda_{1k},...,\lambda_{Nk}]^{T}$, sao cho $\boldsymbol{x}^{*}$ là nghiệm tối ưu cho bài toán sau, tức là thỏa mãn các điều kiện tối ưu KKT của nó, với  $\boldsymbol{\lambda}=\boldsymbol{\lambda}^{*}$ and $\boldsymbol{t}=\boldsymbol{t}^{*}$
\begin{align} \label{Eq:UEAssociation_Sub_Problem}
& \max\limits_{\boldsymbol{x}} \sum\limits_{k \in \mathcal{K}}\sum\limits_{n \in \mathcal{N}}x_{nk}\left[ \left( 1 - q \xi_{k} \right)t_{nk} - \sum\limits_{i \in \mathcal{N}}\lambda_{ik}t_{ik} \right]
\\
& \emph{s.t.} 
\quad (\emph{C1}), (\emph{C2}), (\emph{C3}) \nonumber.
\end{align}
Nghiệm tối ưu $\boldsymbol{x}^{*}$ thỏa mãn hệ phương trình sau cho  $\boldsymbol{\lambda}=\boldsymbol{\lambda}^{*}$ and $\boldsymbol{t}=\boldsymbol{t}^{*}$
\begin{align}
& 
\lambda_{nk} = \left( 1 - q \xi_{k} \right)x_{nk}/y_{k} \label{Eq:lambda_Sub_Problem} \\
&
t_{nk} = r_{nk}/y_{k}.	\label{Eq:t_Sub_Problem}
\end{align}
Ngược lại, nếu  $\boldsymbol{x}^{*}$ là nghiệm tối ưu của bài toán~\eqref{Eq:UEAssociation_Sub_Problem} và thỏa mãn hệ phương trình ~\eqref{Eq:lambda_Sub_Problem} và~\eqref{Eq:t_Sub_Problem} với $\boldsymbol{\lambda}=\boldsymbol{\lambda}^{*}$ và $\boldsymbol{t}=\boldsymbol{t}^{*}$, $\left(\boldsymbol{x}^{*}, \boldsymbol{t}^{*}\right)$ là nghiệm tối ưu bài toán ~\eqref{Eq:UEAssociation_Aux_Problem} với $\boldsymbol{\lambda}=\boldsymbol{\lambda}^{*}$ là biến kép liên quan đến ràng buộc cuối cùng.
\end{theorem}

\begin{proof}
Định lý sẽ được chứng minh thêm ở phần Phụ lục.
\end{proof}

Quan sát thấy rằng vấn đề ~\eqref{Eq:UEAssociation_Sub_Problem} được tối đa hóa bằng cách kết nối mỗi UE với BS, cung cấp tiện ích cao nhất cho UE, tức là., $k^{*} = \max\nolimits_{k \in \mathcal{K}}\left\lbrace \left( 1 - q \xi_{k} \right)t_{nk} - \sum\nolimits_{i \in \mathcal{N}}\lambda_{ik}t_{ik} \right\rbrace, \forall n \in \mathcal{N}$. Tuy nhiên, mỗi BS được yêu cầu phục vụ ít nhất một UE. Do đó, dựa trên Định lý \ref{Theo:UEAssociation_Sub_Problem}, các tác giả đề xuất một thuật toán \textbf{heuristic} cho bài toán liên kết người dùng. đầu tiên, mỗi người dùng UE tính toán tiện ích của nó $\left( 1 - q \xi_{k} \right)t_{nk} - \sum\nolimits_{i \in \mathcal{N}}\lambda_{ik}t_{ik}, \forall k \in \mathcal{K}$ và phổ biến đến tập các BS. Sau đó, mỗi BS tuần tự chọn một trong số $N$ UEs có tiện ích cao nhất. Cuối cùng, mỗi UE không liên kết kết nối với BS mà nó có tiện ích cao nhất. Quy trình này được lặp lại cho đến khi hội tụ, trong đó $\boldsymbol{t}$ và $\boldsymbol{\lambda}$ được cập nhật thông qua phương pháp giống Newton. 

\subsection{Kiểm soát công suất và cố định các liên kết người dùng}
Đối với liên kết người dùng cố định, vấn đề kiểm soát nguồn là
\begin{align} \label{Eq:PowerControl_Problem}
& \max\limits_{\boldsymbol{p}} \sum\limits_{k \in \mathcal{K}}\sum\limits_{n \in \mathcal{N}}\frac{1 - q\xi_{k}}{y_{k}} x_{nk}r_{nk} - q\sum\limits_{k \in \mathcal{K}}\varrho_{k}p_{k}
\\
& \text{s.t.} 
\quad 0 \leq p_{k} \leq p_{k}^{\max}, \; \forall k \in \mathcal{K}. \nonumber
\end{align}
Vì \eqref{Eq:PowerControl_Problem} là một bài toán NP-Hard \cite{Luo2008}; do đó, một thuật toán với thời gian đa thức để tìm giải pháp tối ưu là không thể thực hiện được và giải pháp tối ưu toàn cục phải được tìm ra bằng cách sử dụng các phương pháp tối ưu hóa toàn cục. Trong phần tiếp theo, bài toán \eqref{Eq:PowerControl_Problem} được xấp xỉ thành một chuỗi chương trình lồi bằng cách sử dụng phương pháp SCALE được đề xuất trong \cite{John2009}. Cụ thể, $\alpha\log(z) + \beta \leq \log(1 + z)$, hội tụ tại $z = \tilde{z}$ khi các hệ số gần đúng là $\alpha = \tilde{z}/(1 + \tilde{z})$ và $\beta = \log(1 + \tilde{z}) -  \tilde{z}(1 + \tilde{z})^{-1}\log(\tilde{z})$. Sử dụng các biến biến phụ $\boldsymbol{\rho}$ sao cho $\rho_{k} = \log(p_{k})$ và áp dụng phương pháp SCALE, bài toán \eqref{Eq:PowerControl_Problem} có thể được viết lại như sau:
\begin{align} \label{Eq:PowerControl_Problem_SCALE}
& \max\limits_{\boldsymbol{\rho}} \sum\limits_{k \in \mathcal{K}}\sum\limits_{n \in \mathcal{N}}\frac{1 - q\xi_{k}}{y_{k}} x_{nk}\tilde{r}_{nk}(e^{\boldsymbol{\rho}}) - q\sum\limits_{k \in \mathcal{K}}\varrho_{k}e^{\rho_{k}}
\\
& \text{s.t.} 
\quad \rho_{k} \leq \log(p_{k}^{\max}), \; \forall k \in \mathcal{K}, \nonumber
\end{align}
trong đó $\tilde{r}_{nk}(e^{\boldsymbol{\rho}}) = \alpha_{nk}\log(\gamma_{nk}(e^{\boldsymbol{\rho}})) + \beta_{nk}$.
Bây giờ, \eqref{Eq:PowerControl_Problem_SCALE} là một bài toán tối ưu hóa lồi; do đó, chúng ta có thể sử dụng bất kỳ bộ giải lồi nào có sẵn để giải \cite{Cam2004}. 

Các tác giả đề xuất sử dụng phương pháp "descent" - phương pháp lấy gốc xen kẽ - để giải bài toán trừ \eqref{Eq:SubEq_OptProblem}. Cụ thể, phương pháp rút gốc luân phiên phân tách vấn đề \eqref{Eq:SubEq_OptProblem} thành hai vấn đề riêng biệt: liên kết người dùng và kiểm soát quyền lực, đồng thời giải quyết từng vấn đề một cách cô lập trong cùng một khoảng thời gian. Quá trình này được lặp lại cho đến khi hội tụ như được tóm tắt trong Thuật toán \ref{Alg:Algorithm}. Trong phần tiếp theo, các tác giả sẽ cung cấp phân tích độ hội tụ và độ phức tạp của thuật toán được đề xuất.

Ngoài ra, các phương pháp học máy (Machine learning) dùng để tối ưu hàm mục tiêu cũng được tác giả cân nhắc, tính toán như phương pháp Neural Network, phương pháp hồi quy tuyến tính (Linear Regression). Tuy nhiên, sau khi triển khai một số thử nghiệm, mô phỏng trên mạng truyển thông ảo, tác giả nhận ra vấn đề áp dụng các thuật toán học máy trên thiết bị của người dùng và trên các trạm phát tiêu tốn một lượng lớn điện năng. Đặc biệt, trong môi trường thực tế, việc tính toán, tìm ra giải pháp tối ưu cho việc phân bổ người dùng càng phức tạp và khó khăn khi các thiết bị di động di chuyển với tốc độ cao, trong phạm vi lớn. Vì thế, nhóm tác giả cho rằng các phương pháp học máy là chưa phù hợp với yêu cầu của nghiên cứu.

Cuối cùng, các tác giả quyết định lựa chọn phương pháp "descent" để tối ưu bài toán giảm thiểu công suất của mạng HetNets. Thuật toán được đề xuất dưới dạng lập trình cơ bản như Thuật toán \ref{Alg:Algorithm}:
\begin{algorithm} 
\caption{Thuật toán đề xuất cho bài toán tối ưu công suất trong mạng không đống nhất}\label{Alg:Algorithm}
    \begin{algorithmic}[1]
    \State Set $\varepsilon > 0$, $t = 0$, $\text{FLAG} = 0$, và $q^{(t)} = 0$.
    \Repeat {}
    \State Giải quyết vấn đề liên kết người dùng để phân bổ công suất cố định để thu được $\boldsymbol{x}^{(t)}$.
    \State Giải quyết điều khiển công suất cho một giả định người dùng cố định để có được $\boldsymbol{p}^{(t)}$.
    \If {$\left| R(\boldsymbol{p}^{(t)},\boldsymbol{x}^{(t)}) - q^{(t)}P(\boldsymbol{p}^{(t)},\boldsymbol{x}^{(t)}) \right| \leq \varepsilon$}
    \State Giá trị tối ưu là $\boldsymbol{p}^{*} = \boldsymbol{p}^{(t)}$ và $\boldsymbol{p}^{*} = \boldsymbol{p}^{(t)}$.
    \State Hiệu quả công suất tối ưu là $\boldsymbol{q}^{*} = \boldsymbol{q}^{(t)}$.
    \State Gán $\text{FLAG} = 1$.
    \Else
    \State Tăng $t = t + 1$.
    \State Gán $q^{(t)} = R(\boldsymbol{p}^{(t-1)},\boldsymbol{x}^{(t-1)})/P(\boldsymbol{p}^{(t-1)},\boldsymbol{x}^{(t-1)})$.
    \EndIf
    \Until{$\text{FLAG} = 1$}
    \end{algorithmic}
\end{algorithm}

\subsection{Phân tích sự hội tụ và độ phức tạp của kết quả thu được}
Ở mỗi lần lặp lại của thuật toán được đề xuất, các tác giả giải quyết lần lượt một vấn đề liên kết người dùng và một vấn đề điều khiển công suất. Do đó, để chứng minh sự hội tụ của thuật toán đề xuất, cần phải chứng minh sự hội tụ của hai bài toán con bên trên. Nhóm tác giả trình bày hai định lý là Định lý 3 và Định lý 4. Ngoài ra, giá trị tối ưu $\boldsymbol{p}^{(*)}$ sẽ thu được bằng thuật toán Dinkelbach’s với đảm bảo hội tụ \cite{Book1967}. Do đó, nhóm tác giả kết luận rằng thuật toán được đề xuất trong nghiên cứu là hội tụ.

\begin{theorem}
Đối với $q^{(t)}$ và $\boldsymbol{p}^{(t)}$, cố định, thủ tục giải bài toán liên kết người dùng được đảm bảo hội tụ \cite{He2014}.
\end{theorem}

\begin{theorem}
Với phép xấp xỉ logarit, đối với $q^{(t)}$ và $\boldsymbol{x}^{(t+1)}$ cố định, lời giải cho bài toán xấp xỉ \eqref{Eq:PowerControl_Problem_SCALE} đơn điệu cải thiện hàm mục tiêu của nó. Ngoài ra, lời giải cho \eqref{Eq:PowerControl_Problem_SCALE} hội tụ về một điểm, thỏa mãn điều kiện tối ưu KKT của bài toán \eqref{Eq:PowerControl_Problem}, đã được chứng minh ở bài báo \cite{QVPFairness}.
\end{theorem}

Giả sử rằng chúng ta sử dụng phương pháp đối ngẫu để giải quyết vấn đề kiểm soát công suất \eqref{Eq:PowerControl_Problem}. Gọi $T_{2}$ là số lần lặp cần thiết để cập nhật các hệ số xấp xỉ $\boldsymbol{\alpha}$ và $\boldsymbol{\beta}$. Đòng thời, $L$ là số lần lặp để giải bài toán \eqref{Eq:PowerControl_Problem_SCALE} trong miền đối ngẫu. Khi đó, độ phức tạp tính toán của mỗi bước để giải quyết vấn đề điều khiển công suất là $\mathcal{O}\left( T_{2}LK \right)$. Độ phức tạp cần thiết để giải bài toán liên kết người dùng \eqref{Eq:UEAssociation_Problem} là $\mathcal{O}\left( (m+1)T_{1}N(N-K) \right)$, trong đó $T_{1}$ là số lần lặp cần thiết để giải bài toán \eqref{Eq:UEAssociation_Sub_Problem}, $N-K = 1$ nếu $N = K$, tức là, số lượng BS không nhỏ hơn số lượng UE, và $m$ là một số nhỏ được sử dụng để cập nhật $\boldsymbol{t}$ và $\boldsymbol{\lambda}$. Gọi $T_{3}$ là số lần lặp cần thiết để cập nhật $q$, bước $11$ trong Thuật toán \ref{Alg:Algorithm}. Khi đó, giá trị lớn nhất giữa $\mathcal{O}\left( (m+1)T_{1}N(N-K)T_{3} \right)$ và $\mathcal{O}\left( T_{2}LKT_{3} \right)$ là độ phức tạp tính toán của thuật toán đề xuất.

\section{Kết quả mô phỏng}
\label{Sec:SimulationResults}

Xem xét một mạng không dây hỗn hợp hai tầng, trong đó có $8$ trạm phát sóng trong nhà - "femtocell" được triển khai ngẫu nhiên trong vùng bao phủ của cell cỡ lớn là $500$(m) $\times$ $500$(m) trung tâm là MBS. Nhóm tác giả giả định rằng $240$ thiết bị di động được phân phối ngẫu nhiên trên vùng phủ của cell cỡ lớn và cấu trúc liên kết mạng được cố định trong thời gian mô phỏng. Các tác giả cũng giả định rằng khoảng cách tối thiểu giữa MBS và UE là $35$m, giữa các BS nhỏ và UE là $10$ m, và giữa hai UE bất kỳ là $3$m. Công suất phát lớn nhất của MBS và BS nhỏ lần lượt là $46$dBm và $30$dBm, công suất tiêu thụ tĩnh của MBS và BS nhỏ lần lượt là $10$W và $0.1$W, nghịch đảo của hiệu suất tiêu của MBS và BS nhỏ lần lượt là $4$ và $2$, và thông số công suất tiêu thụ $\xi_{k}$ là 1 W/Mbps cho tất cả các liên kết trạm lõi. Băng thông của hệ thống là 10 MHz và công suất tiếng ồn nền $n_{0}$ là $-104$ dBm. Mô hình suy hao đường dẫn cho MBS là $g_{nk} = 128.1 + 37.6\log_{10}(d_{nk})$ và cho các BS nhỏ là $g_{nk} = 140.7 + 36.7\log_{10}(d_{nk})$, trong đó $d_{nk}$ là khoảng cách. Cuối cùng tác giả giả định rằng mức tăng công suất kênh bao gồm suy hao đường dẫn và mờ dần bóng mờ, có độ lệch chuẩn $8$ dB. 
\begin{figure} [ht]
    \centering
    \includegraphics[scale=0.35]{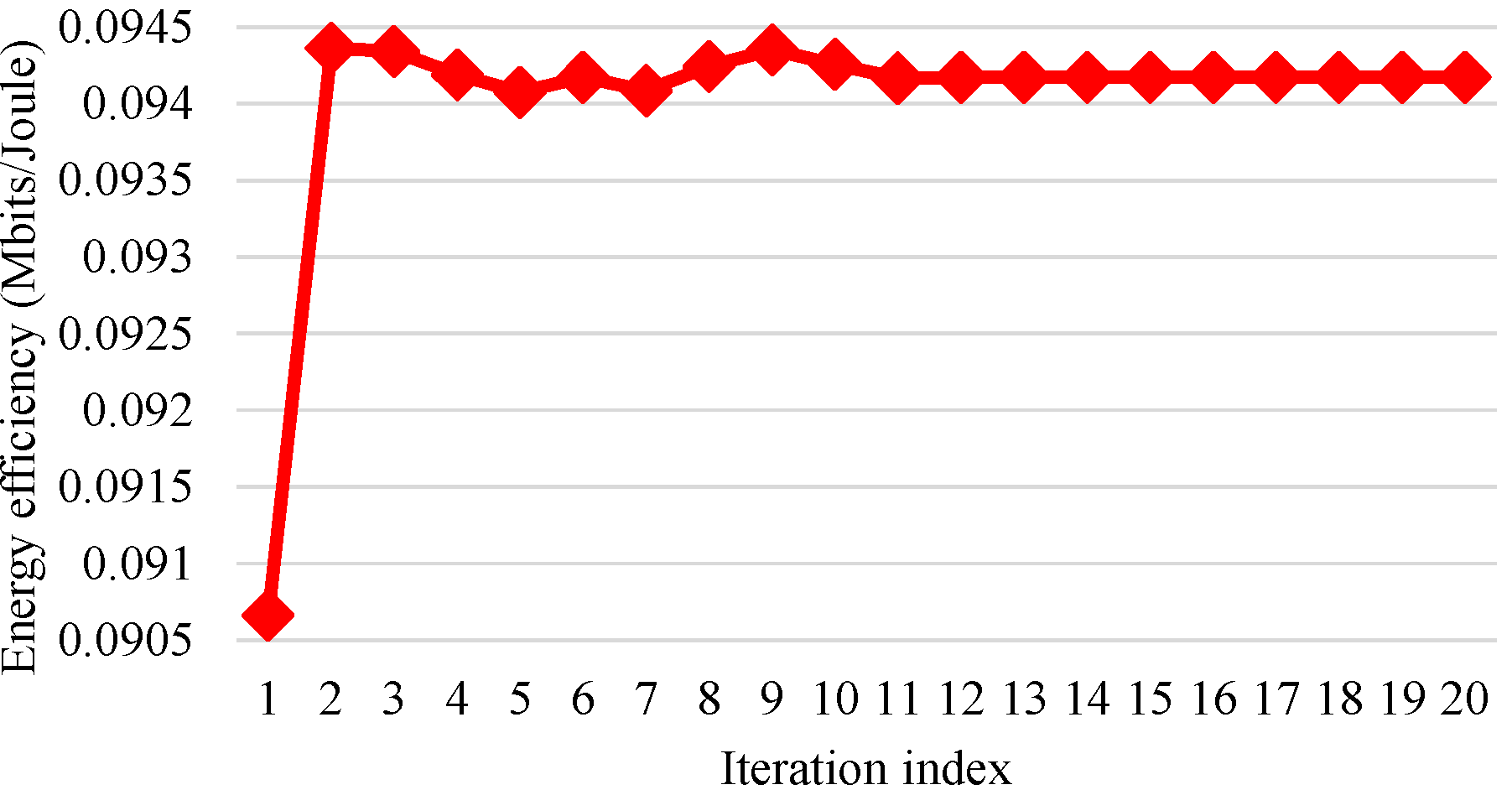}
    \caption{Sự hội tụ của thuật toán UAPCEE.}
    \label{Fig:Convergence}
\end{figure}

Đầu tiên, tác giả chỉ ra sự hội tụ của thuật toán được đề xuất (User Association and Power Control for Energy Efficiency - UAPCEE). Theo quan sát từ Hình \ref{Fig:Convergence}, thuật toán được đề xuất có thể hội tụ với mười một lần lặp lại, khi đó giải pháp là tối ưu. Hiệu suất sử dụng mạng cũng gia tăng đáng kể, đem lại kết quả tối ưu hơn. Tiếp theo, nhóm tác giả so sánh các mô hình quản lý nhiễu TOIM khác nhau về hiệu suất công suất với các thông số tiêu thụ điện trên một đơn vị dữ liệu như trong Bảng \ref{tab:compare}.

\begin{table}[t]
    \centering
    \caption{So sánh hiệu suất công suất của một số mô hình hiện có}
    \begin{tabular}{|m{0.15\columnwidth}|m{0.14\columnwidth}|m{0.18\columnwidth}|m{0.16\columnwidth}|m{0.12\columnwidth}|}
        \hline
        Mức điện năng tiêu thụ  (W/Mbps)
        &UAPCEE  (Mbps/ Joule)
        &JUAPCMSER (Mbps/Joule)
        &UAPCEEwB (Mbps /Joule)
        &RE (Mbps/Joule) \\
        \hline
        \centering 8    &0.116    &0.1114    &0.1111    &0.113  \\
        \hline
        \centering 10   &0.0941   &0.0911    &0.0909    &0.0921  \\
        \hline
        \centering 12   &0.0789   &0.0771    &0.0769    &0.0778  \\
        \hline
        \centering 14   &0.0679   &0.0668    &0.0667    &0.0673  \\
        \hline
        \centering 16   &0.0599   &0.0589    &0.0588    &0.0593\\
        \hline    
        \centering 18   &0.0536   &0.0527    &0.0526    &0.0531\\
        \hline
        \centering 20   &0.0484   &0.0477    &0.0476    &0.0486\\
        \hline
    \end{tabular}
    \label{tab:compare}
\end{table}

Các mô hình được so sánh bao gồm: (1) Mở rộng phạm vi (Range Expansion - RE); (2) Người dùng kết hợp với BS có độ lợi kênh lớn nhất; (3) Không tính đến công suất tiêu thụ của các liên kết backhaul (UAPCEEw) và (4) Không xem xét công suất tiêu thụ của các liên kết backhaul (JUAPCMSER).

Hình \ref{Fig:Comparison} cho thấy hiệu suất năng lượng và mức chênh lệch về hiệu suất giữa các thuật toán giảm khi tham số động  tăng lên. Điều đó là hợp lý bởi vì tiêu thụ điện năng trong các liên kết backhaul trở nên chiếm ưu thế so với tiêu thụ điện năng trong mạng truy nhập. Ngoài ra, thuật toán được đề xuất vượt trội hơn so với các khuôn khổ được so sánh vì công suất tiêu thụ của các liên kết hỗ trợ được tính đến hàm mục tiêu của hiệu quả năng lượng.
\begin{figure} [ht]
    \centering
    \includegraphics[scale=0.35]{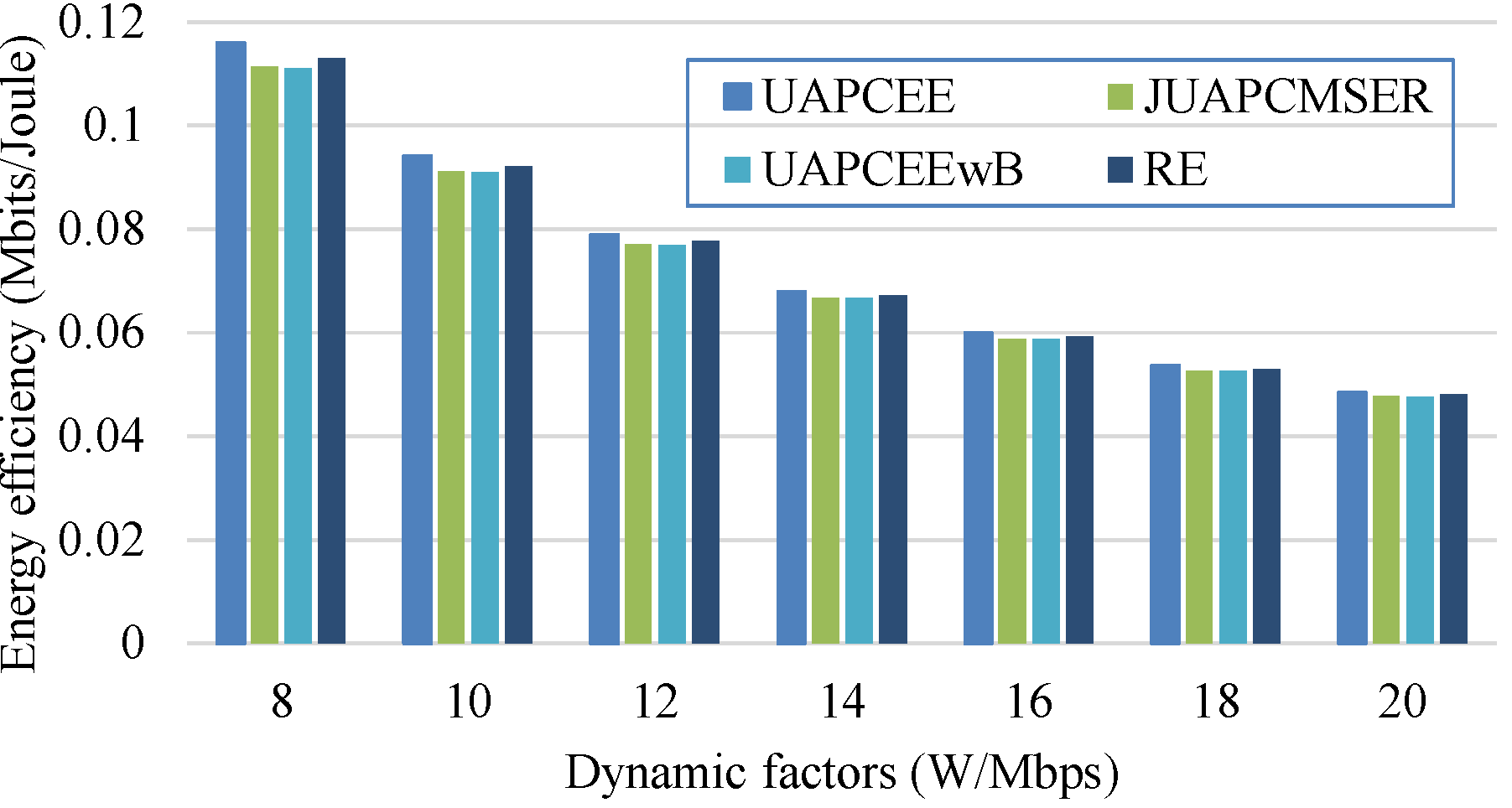}
    \caption{Tương quan hiệu suất năng lượng của các mô hình hiện nay.}
    \label{Fig:Comparison}
\end{figure}

\section{Tổng kết}
Tối ưu công suất tiêu thụ điện trong mạng không đồng nhất là một vấn đề rất quan trọng trong mạng truyền thông, viễn thông ngày nay. Việc giảm thiểu tối đa công suất cần thiết của các thiết bị phát cũng như thiết bị sử dụng mạng sẽ mang lại nhiều kết quả, đổi mới trong tương lai và đặc biệt hữu ích đối với các mạng IoT, mạng viễn thông nội đô. Trong nghiên cứu, nhóm tác giả đã triển khai một mô hình liên kết người dùng và điều khiển công suất để tối ưu năng lượng trong không dây không đồng nhất mật độ cao. Vì vấn đề ban đầu là rất khó giải quyết, nhóm tác giả đã đề xuất một thuật toán lặp với đảm bảo sự hội tụ. Từ kết quả mô phỏng được cung cấp, nhóm tác giả đã đưa ra điểm tiến bộ của thuật toán đề xuất so với các phương pháp hiện có. Như đã nói ở trên, công trình này được thiết kế cho các mạng đơn ăng ten, việc mở rộng cho mạng nhiều ăng ten là rất cần thiết. Ngoài ra, công việc trong tương lai nên tính đến các vấn đề liên quan đến những hành vi không lý tưởng.

\appendix
Gọi $\boldsymbol{\lambda}$ vectơ đối ngẫu liên kết với ràng buộc cuối cùng trong~\eqref{Eq:UEAssociation_Aux_Problem}, hàm Lagrangian được định nghĩa là 
\begin{align*}
L(\boldsymbol{x},\boldsymbol{t},\boldsymbol{\lambda}) = 
&
\sum\limits_{k \in \mathcal{K}}\sum\limits_{n \in \mathcal{N}}\left( 1 - q\xi_{k} \right) x_{nk}t_{nk} \\
&
+ \sum\limits_{k \in \mathcal{K}}\sum\limits_{n \in \mathcal{N}}\lambda_{nk}\left( r_{nk} - y_{k}t_{nk} \right).
\end{align*}

Vì $\boldsymbol{\lambda}$ là vectơ đối ngẫu và $(\boldsymbol{x}^{*},\boldsymbol{t}^{*})$ là nghiệm tối ưu cho bài toán~\eqref{Eq:UEAssociation_Sub_Problem}, các phương trình sau có thể được tương đương với: 
\begin{align*}
&
\partial L/\partial t_{nk} = \left( 1 - q \xi_{k} \right)x_{nk}^{*} - \lambda_{nk}^{*}y_{k}^{*} = 0,\\
&
\lambda_{nk}^{*}\left( r_{nk} - y_{k}^{*}t_{nk}^{*} \right) = 0.
\end{align*}

Hệ phương trình trên tương đương với $\lambda_{nk}^{*} = \left( 1 - q \xi_{k} \right)x_{nk}^{*}/y_{k}^{*}$ và $t_{nk}^{*} = r_{nk}/y_{k}^{*}$. Ngoài ra, đối với $\boldsymbol{\lambda}=\boldsymbol{\lambda}^{*}$ và $\boldsymbol{t}=\boldsymbol{t}^{*}$, bài toán sau
\begin{align*} \label{Eq:UEAssociation_Sub_Problem_Appendix}
& \max\sum\limits_{k \in \mathcal{K}}\sum\limits_{n \in \mathcal{N}}\left( 1 - q\xi_{k} \right) x_{nk}t_{nk} + \sum\limits_{k \in \mathcal{K}}\sum\limits_{n \in \mathcal{N}}\lambda_{nk}\left( r_{nk} - y_{k}t_{nk} \right)
\\
& \text{s.t.} 
\quad (\text{C1}), (\text{C2}), (\text{C3})
\end{align*}
có thể được đơn giản hóa thành vấn đề~\eqref{Eq:UEAssociation_Sub_Problem} and và tất cả chúng đều có cùng một tập hợp các điều kiện tối ưu KKT. Điều này hoàn thành bằng chứng về câu lệnh đầu tiên của Định lý~\ref{Theo:UEAssociation_Sub_Problem}. Khẳng định thứ hai có thể được chứng minh tương tự. Phần chứng minh kết thúc.

\bibliographystyle{IEEEtran}
\nocite{*}
\def\refname{Tài liệu tham khảo}
\bibliography{reference}

\vspace{12pt}

\end{document}